\documentclass[12pt, draftclsnofoot, onecolumn]{IEEEtran}

\usepackage{amsmath}
\usepackage{amssymb}
\usepackage{amsthm}
\usepackage{cite}

\usepackage{graphicx}
\usepackage{enumerate}
\usepackage{setspace}
\usepackage{subfigure}
\usepackage{color}
\usepackage{multirow}
\usepackage{slashbox}
\usepackage{accents}
\newlength{\dhatheight}
\newcommand{\doublehat}[1]{%
    \settoheight{\dhatheight}{\ensuremath{\hat{#1}}}%
    \addtolength{\dhatheight}{-0.35ex}%
    \hat{\vphantom{\rule{1pt}{\dhatheight}}%
    \smash{\hat{#1}}}}

\newtheorem{theorem}{Theorem}
\newtheorem{definition}{Definition}
\newtheorem{lemma}{Lemma}

\title{Outage-based ergodic link adaptation\\ for fading channels with delayed CSIT}

\author{Jung Hyun Bae, Jungwon Lee, Inyup Kang\\
Mobile Solutions Lab\\
Samsung US R$\&$D Center\\
San Diego, CA, USA\\
Email: jungbae@umich.edu, jungwon@alumni.stanford.edu, inyup.kang@samsung.com}% Change to the current month of the series
\begin{document}
\maketitle
\begin{abstract}
Link adaptation in which the transmission data rate is dynamically adjusted according to channel variation is often used to deal with time-varying nature of wireless channel. When channel state information at the transmitter (CSIT) is delayed by more than channel coherence time due to feedback delay, however, the effect of link adaptation can possibly be taken away if this delay is not taken into account. One way to deal with such delay is to predict current channel quality given available observation, but this would inevitably result in prediction error. In this paper, an algorithm with different view point is proposed. By using conditional cdf of current channel given observation, outage probability can be computed for each value of transmission rate $R$. By assuming that the transmission block error rate (BLER) is dominated by outage probability, the expected throughput can also be computed, and $R$ can be determined to maximize it. The proposed scheme is designed to be optimal if channel has ergodicity, and it is shown to considerably outperform conventional schemes in a Rayleigh fading channel model. 
\end{abstract}
% This first line makes the cover page, which prints the TR number.

% This second line makes the title portion of the first page.
%\maketitle

%\listoffigures
%\listoftables
\section{Introduction}
\label{sec:intro}
A time-varying nature is one of the most important properties of wireless channel. To ensure efficient and reliable communication in time-varying fading channel, the current wireless standards support link adaptation in which the transmission data rate is dynamically adjusted according to channel variation~\cite[Ch. 10]{SeToBa09}. An apparent challenge with link adaptation is delay of channel state information at the transmitter (CSIT). If the receiver observes channel at some point and feeds it back to the transmitter, then CSIT would inevitably suffer from delay which can be several transmission blocks.\footnote{In current wireless standards~\cite[Ch. 10]{SeToBa09}, rate determination is done at the receiver side, and the rate request to the transmitter is delayed. In this paper, we assume that rate determination is done at the transmitter side with delayed CSIT. Note that these two scenarios are equivalent.} Unless channel coherence time is larger than this delay, CSIT delay can result in incorrect CSIT which possibly take away most of benefits of link adaptation. \\
\indent
One way of dealing with CSIT delay is to predict current channel quality given delayed observation (feedback). When channel statistics is known, the optimal predictor which minimizes mean-square-error can be computed by using conditional mean~\cite[Thm. 3.A.1]{KaSa00}. In practice, linear predictors which are also optimal when channel is Gaussian distributed are often used, and \cite{ToQu08, CuLu10, WaZh11} show several examples of this approach. One thing to note is that prediction of current channel can never be perfect due to the random nature of current channel, i.e., current channel is a random variable conditioned on observation. With the approach of predicting current channel, this random nature is source of incorrect CSIT and is undesirable but unavoidable. \\
\indent
In this paper, we propose a scheme with completely different view point from predicting channel quality. In essence, the proposed scheme finds the transmission data rate which maximizes the expected throughput by using the cdf of current channel conditioned on observation. Since it does not attempt to predict channel quality, it is obviously free from error of such prediction. This does not mean that every transmission of the proposed scheme would result in successful decoding since the determined data rate may correspond to non-zero error rate of the transmitted code block, i.e., block error rate (BLER). Instead, it determines the `best' BLER which maximizes the throughput. If ergodicity holds for a random process defined by joint distribution of channels, then the proposed scheme maximizes the actual long-term throughput, and hence, it is optimal over all possible link adaptation schemes. Such assumption of ergodicity is thought to hold in many models as discussed in~\cite[Ch. 5]{TsVi05}, and numerical evaluation for a specific model will be seen later in this paper. Although performance of the proposed scheme is guaranteed by the aforementioned optimality, it will also be seen that simulation results show that the proposed scheme significantly outperforms conventional ones.\\
\indent 
An important assumption used to derive the proposed scheme is that code block error is resulted from an outage event. An outage event in which the channel quality does not support the transmission data rate is often used to analyze performance in slow fading channels with no CSIT~\cite[Ch. 5]{TsVi05}, and outage probability is determined by the transmission data rate and the cdf of channel. The situation with delayed CSIT is very similar to no CSIT case, and the only difference is that the cdf of channel varies depending on observation. If channel coding is optimal with correct CSIT, then the expected BLER must be determined by outage probability with delayed CSIT. Since this is not true in any practical scenario, outage probability alone does not determine the expected BLER, but such assumption is still quite realistic given the fact that practical channel coding is near optimal and major source of error is lack of CSIT rather than sub-optimal coding when CSIT is delayed. Once outage probability is found, the expected throughput of the transmission data rate can also be found by using the above assumption of the expected BLER being determined by outage probability, and hence, the rate which maximizes the throughput can be determined. \\
\indent
The remainder of the paper is organized as follows. Section~\ref{sec:ch} describes channel model used in this paper. Section~\ref{sec:out} describes the proposed scheme in detail and presents performance evaluation as well as numerical evaluation of ergodicity which is closely related to optimality of the proposed scheme. Section~\ref{sec:ext} discusses possible generalizations of the proposed scheme to more complicated cases, and Section~\ref{sec:con} concludes the paper.
%%%%%%%%%%%%%%%%%%%%%%%%%%%%%%%
%%%%%%%%%%%%%%%%%%%%%%%%%
\section{Channel model}
\label{sec:ch}
Consider the following single-input single-output (SISO) block flat fading channel model corresponding to the $n$th symbol of the $i$th code block with channel output $y_i[n]$ at the receiver, channel input $x_i[n]$ from the transmitter, background noise $z_i[n]$, and the channel $h_i$. 
\begin{equation}
\label{eq:ch}
y_i[n] = \sqrt{P}h_ix_i[n] + z_i[n] \text{ for } n = 1,...,N,
\end{equation}
where $N$ is the code block length
We assume that background noise $z_i[n] \sim \mathcal{CN}(0,1)$ is iid. $\{h_i\}_{i=1}^\infty$ is Rayleigh fading process which means that $h_i\sim  \mathcal{CN}(0,1)$, and $P$ represents signal-to-noise ratio (SNR). We assume that $\{h_i\}_{i=1}^\infty$ is jointly Gaussian. One example of such process can be found in Clarke's model~\cite[Ch. 2]{TsVi05} which satisfies $E[h_ih^*_{i+l}]=J_0(2\pi f_D  lT_b  )$ for all integer $l$ where $f_D$ is Doppler frequency, $T_b$ is code block duration, and $J_0$ is a zeroth-order Bessel function of the first kind.\\
\indent
For the code block $i$, we assume that the maximum achievable transmission data rate is the same as the capacity of $h_i$ which is 
\begin{equation}
C_i= \log_2 (1+P|h_{i}|^2 ).
\end{equation}
Note that this assumption is not true unless length $N$ of each code block approaches infinity. For the ease of exposition, however, we will keep this assumption throughout the paper.  We also assume that CSIT delay is $l_d$ code block duration throughout the paper, i.e., CSIT available at the $i$th code block corresponds to the channel at the $(i-l_d)$th code block.\\
\indent
Although we are assuming one of the simplest time-varying fading channel model here, the proposed scheme can be generalized to more complicated cases such as with multiple antennas and/or with fading within a code block. These generalizations will be discussed in section~\ref{sec:ext}.  
%%%%%%%%%%%%%%%%%%%%%%%%%
%%%%%%%%%%%%%%%%%%%%%%%%%%%%%
\section{Outage-based ergodic link adaptation}
\label{sec:out}
\subsection{Description of the scheme}
\label{sec:al}
Consider the $i$th code block. This code block would result in error if and only if the attempted transmission data rate $R$ is greater than $C_i$, i.e., an outage occurs. As mentioned earlier, $h_i$ can be viewed as a random variable conditioned on available observations when there is CSIT delay, and hence, outage probability can be computed for each value of $R$. This outage probability can be used to compute the expected throughput at the $i$th code block. To compute outage probability, we first need to identify distribution of $h_i$ conditioned on observation. To do that, we need the following lemma. 
\begin{lemma}~\cite[Chap. 3.C]{KaSa00}
\label{lem:gau}
If $[X^T, Y^T]^T$ is a Gaussian random vector with real numbers, then given $Y=y$,
\begin{equation}
X \sim N(E[X|Y=y],C_{X|Y}),
\end{equation}
where $C_{X|Y}=C_X-AC_{YX}$, $E[X|Y=y]=A(y-m_Y)+m_X$, $A$ solves $AC_Y=C_{XY}$, $m_X=E[X]$, $C_{XY}=E[(X-m_x)(Y-m_Y)^T]$, and $C_X=C_{XX}$. 
\end{lemma}
Let $h_{obs(i)}$ be the vector of channel observations used to determine the rate for the $i$th code block. Let $a_{re}$ be real part of $a$, and $a_{im}$ be imaginary part of $a$ for a complex scalar or vector $a$. Then, we can set $X=(h_{i,re},h_{i,im})^T$ and $Y=(h^T_{obs(i),re},h^T_{obs(i),im})^T$, and apply Lemma~\ref{lem:gau} to find the conditional distribution of $h_i$. To simplify analysis, we consider $h_{obs(i)}=h_{i-l_d}$ which means only single observation is used. Note that the number of channel observations in $h_{obs(i)}$ only affects calculation of $E[X|Y=y]$ and $C_{X|Y}$, and the following analysis directly carries over to any number of observations in $h_{obs(i)}$. \\
\indent
With restriction of single observation, the only thing which needs to be specified is correlation between $h_i$ and $h_{i-l_d}$. Let $C=E[h_{i}h^*_{i-l_d}]$. In Clarke's model, $C=J_0(2\pi f_D  l_dT_b  )$. Given $h_{i-l_d}$,
\begin{equation}
h_{i} \sim \mathcal{CN} (Ch_{i-l_d}, 1-C^2).
\end{equation}
Outage probability $P_{out}(R)$ of the $i$th code block for given $R$ can be expressed as 
\begin{subequations}
\begin{eqnarray}
P_{i,out}(R)&=& P \{ C_{i}<R | h_{i-l_d} \}\\
          &=& P\Big\{|h_{i}|^2<\frac{2^R-1}{P}\Big| h_{i-l_d}\Big\}.
\end{eqnarray}
\end{subequations}
Note that $|h_{i}|^2$ given $h_{i-l_d}$ has noncentral chi-squared distribution with 2 degrees of freedom, and its cdf is given as 
\begin{equation}
\label{eq:cdf1}
P\{|h_{i}|^2<x|h_{i-l_d} \}=1-Q_1\bigg(  \sqrt{\frac{2C|h_{i-l_d}|^2}{1-C^2}}, \sqrt{\frac{2x}{1-C^2}}\bigg),
\end{equation}
where $Q_1$ is a Marcum Q-function defined as
\begin{equation}
Q_1(a,b)=\int_b^\infty x\exp\Big(-\frac{x^2+a^2}{2}\Big)I_0(ax)dx,
\end{equation}
where $I_0$ is a modified Bessel function of order 0. Therefore,
\begin{equation}
P_{i,out}(R)=1-Q_1\bigg(  \sqrt{\frac{2C|h_{i-l_d}|^2}{1-C^2}}, \sqrt{\frac{2(2^R-1)}{P(1-C^2)}}\bigg).
\end{equation}
We can also express the expected throughput $TP_i(R)$ of the $i$th code block for given $R$ as
\begin{subequations}
\label{eq:tp}
\begin{eqnarray}
TP_i(R)&=&R(1-P_{i,out}(R))\\
       &=&RQ_1\bigg(  \sqrt{\frac{2C|h_{i-l_d}|^2}{1-C^2}}, \sqrt{\frac{2(2^R-1)}{P(1-C^2)}}\bigg).
\end{eqnarray}
\end{subequations}
The proposed scheme is to choose $R$ which is the solution of $\max_R TP_i(R)$ as the transmission data rate of the $i$th code block. It can be seen that $TP_i(R)$ becomes the same as the case where $|h_i|$ is Rayleigh distributed if $C=0$. In other words, if channel is uncorrelated, then delayed CSIT fails to give any information about current channel. Hence, the proposed scheme as well as conventional methods of predicting channel would be beneficial only if there is considerable correlation in the channel. Performances of the proposed scheme with respect to various channel correlations will be presented later in this paper.\\
\indent 
As mentioned earlier, an important aspect of the proposed scheme is getting rid of the process of predicting current channel. Validity of this approach lies in ergodicity of the channel. It can be easily seen that the transmitter allocates the same rate for code block $i$'s with which $|h_{i-l_d}|$'s are the same. If we only consider the actual throughput of code blocks with those $i$'s, then it must converge to $\max_R TP_i(R)$ given ergodicity of the channel, and the proposed scheme indeed maximizes the actual long-term throughput. 
%%%%%%%%%%%%%%%%%%%%%%%%%%%%%%%%%%%%%%%%%%%%%%%%%%%%%%%%
%%%%%%%%%%%%%%%%%%%%%%%%%%%%%%%%%%%%%%%%%%%%%%%%%%%%%%%%
\subsection{Determination of search interval}
\label{sec:int}
$TP_i(R)$ given in~\eqref{eq:tp} involves Marcum Q-function which does not have a closed form expression, and this implies that the solution of $\max_R TP_i(R)$ cannot be found analytically. In the case when $TP_i(R)$ is convex, well known exiting algorithms can be applied to find the solution of $\max_R TP_i(R)$~\cite{BoVa09}. To the authors' best knowledge, however, $TP_i(R)$ is not known or can be proven to be convex, and hence, we consider brute-force search. Since possible values of $R$ are any positive numbers, the optimal search for the solution of $\max_R TP_i(R)$ can be prohibitively difficult. To reduce complexity of search, we may first transform the original continuous search problem to a discrete search problem by quantizing search space of $R$ which is non-negative real line. This would result in loss of optimality, but even after that, complexity of search would still be prohibitive unless finite search interval is established. In this section, we will discuss of ways to establish such interval.\\
\indent
Since we need to make sure that the solution of the original optimization problem $\max_R TP_i(R)$ belongs to the reduced search interval $[R_L,R_U]$, $R_L$ and $R_U$ must satisfy the following conditions. First, $\frac{dTP(R)}{dR}\geq 0$ for all $R<R_L$. Second, $\frac{dTP(R)}{dR}\leq 0$ for all $R>R_U$. Hence, our objective is to find hopefully large $R_U$ and small $R_L$ satisfying these conditions. The following theorem gives the interval $[R_L,R_U]$ which the solution of $\max_R TP_i(R)$ always belongs to. Let $\alpha=\sqrt{\frac{2C|h_{i-l_d}|^2}{1-C^2}}$, $\beta=\sqrt{\frac{2}{P(1-C^2)}}$ and $\gamma(R)=\sqrt{2^R-1}$.
\begin{theorem}
\label{thm:int}
The solution of $\max_R TP_i(R)$ always belongs to the interval $[R_L,R_U]$ where $R_U=\max\{   \hat{R},\log_2 (\frac{\alpha^2}{\beta^2}+1 )\}$ and $R_L=\doublehat{R}$. $\hat{R}$ is the unique solution of $R2^{R-1}=\frac{1+\alpha\sqrt{\frac{\pi}{2}}}{\beta^2\ln2}$, and $\doublehat{R}$ is the unique solution of $R2^{R-1}=\frac{1}{\beta^2\ln2}$.
\end{theorem}
\begin{proof}
$TP_i(R)$ can be written as follows.
\begin{subequations}
\begin{eqnarray}
TP_i(R)&=& R\int_{\beta\gamma(R)}^\infty x\exp\Big(-\frac{x^2+\alpha^2}{2}\Big)I_0(\alpha x)dx\\
     &=&R\bigg(1-\int_0^{\beta\gamma(R)} x\exp\Big(-\frac{x^2+\alpha^2}{2}\Big)I_0(\alpha x)dx\bigg).
\end{eqnarray}
\end{subequations}
Hence,
\begin{subequations}
\begin{eqnarray}
\frac{dTP_i(R)}{dR}&=& Q_1(\alpha,\beta\gamma(R) )-R\frac{d\beta\gamma(R)}{dR}\beta\gamma(R) \exp\Big(-\frac{\beta^2\gamma(R)^2+\alpha^2}{2}\Big)I_0(\alpha \beta\gamma(R))\\
     &=&Q_1(\alpha,\beta\gamma(R) )-R\frac{\beta 2^R\ln2}{2\gamma(R)}\beta\gamma(R) \exp\Big(-\frac{\beta^2\gamma(R)^2+\alpha^2}{2}\Big)I_0(\alpha \beta\gamma(R))\\
     &=&Q_1(\alpha,\beta\gamma(R) )-R2^{R-1}(\ln2)\beta^2\exp\Big(-\frac{\beta^2\gamma(R)^2+\alpha^2}{2}\Big)I_0(\alpha \beta\gamma(R)).
\end{eqnarray}
\end{subequations}
The above expression still involves $Q_1(\cdot)$ and $I_0(\cdot)$ which do not have closed form expressions. Hence, we consider bounds on these functions to obtain an analytical solution on the search interval.\\
\indent
 To find $R_U$, let us first upper bound $\frac{dTP_i(R)}{dR}$ by upper bounding $Q_1(\cdot)$. Bounds on Marcum Q-function have been studied considerably as seen in~\cite{CoFe02} and references therein. In~\cite{CoFe02}, the upper bound of $Q_1(\alpha,\beta\gamma(R))$ is given for $\beta\gamma(R)>\alpha$ as 
\begin{equation}
Q_1(\alpha,\beta\gamma(R)) \leq I_0(\alpha \beta\gamma(R))\bigg(\exp\Big(-\frac{\beta^2\gamma(R)^2+\alpha^2}{2}\Big)+\frac{\alpha}{\exp(\alpha \beta\gamma(R))}\sqrt{\frac{\pi}{2}}\text{erfc}\Big(\frac{\beta\gamma(R)-\alpha}{\sqrt{2}}\Big)   \bigg),
\end{equation}
which is shown to be quite tight by numerical evaluation in~\cite{CoFe02}. Here, $\text{erfc}(x)=\frac{2}{\sqrt{\pi}}\int_x^\infty e^{-t^2} dt,$\\$x\geq 0$, and the fact that it does not have a closed form expression again becomes problematic. To proceed further, we need to upper bound $\text{erfc}(x)$. To simplify analysis, we consider a upper bound of $\text{erfc}(x)$ which consists of a single exponential term. Well known Chernoff bound which is given as $\text{erfc}(x)\leq 2\exp(-x^2)$~\cite{ChCoMi11} is an example of such bound.  In~\cite{ChCoMi11}, it is shown that $\exp(-x^2)$ is the best upper bound of $\text{erfc}(x)$ among general single-term exponential-type bounds $a\exp(-bx^2)$, where $a,b>0$. By using this, we get
\begin{equation}
Q_1(\alpha,\beta\gamma(R)) \leq I_0(\alpha \beta\gamma(R))\bigg(\exp\Big(-\frac{\beta^2\gamma(R)^2+\alpha^2}{2}\Big)+\alpha\sqrt{\frac{\pi}{2}}\exp\Big(-\frac{\beta^2\gamma(R)^2+\alpha^2}{2}\Big)   \bigg).
\end{equation}
By upper bounding $Q_1(\alpha,\beta\gamma(R))$, we get
\begin{equation}
\frac{dTP_i(R)}{dR} \leq I_0(\alpha \beta\gamma(R))\exp\Big(-\frac{\beta^2\gamma(R)^2+\alpha^2}{2}\Big)\bigg(1+\alpha\sqrt{\frac{\pi}{2}}-R2^{R-1}\beta^2\ln2 \bigg).
\end{equation}
It can be seen that $R_U$ can be found using the above expression without further bounding $I_0(\cdot)$. Indeed, we can say that $\frac{dTP(R)}{dR}\leq 0$ if $R\geq \hat{R}$ where $\hat{R}$ is the unique solution of $R2^{R-1}=\frac{1+\alpha\sqrt{\frac{\pi}{2}}}{\beta^2\ln2}$. Then, $R_U=\max\{   \hat{R},\log_2 (\frac{\alpha^2}{\beta^2}+1 )\}$.\\
\indent
To find $R_L$, we need to lower bound $TP_i(R)$ by lower bounding $Q_1(\alpha,\beta\gamma(R))$. A tight lower bound is also derived in~\cite{CoFe02}, but its complicated expression makes it hard to analyze, and hence, we use another lower bound given in~\cite{CoFe02} which turns out to be looser by numerical evaluation in~\cite{CoFe02}. The lower bound has the following expression.
\begin{equation}
Q_1(\alpha,\beta\gamma(R))\geq I_0(\alpha \beta\gamma(R))\exp\Big(-\frac{\beta^2\gamma(R)^2+\alpha^2}{2}\Big).
\end{equation}
Then, 
\begin{equation}
\frac{dTP_i(R)}{dR} \geq I_0(\alpha \beta\gamma(R))\exp\Big(-\frac{\beta^2\gamma(R)^2+\alpha^2}{2}\Big)\bigg(1-R2^{R-1}\beta^2\ln2 \bigg).
\end{equation}
We can say that $\frac{dTP_i(R)}{dR}\geq 0$ if $R\leq \doublehat{R}$ where $\doublehat{R}$ is the unique solution of $R2^{R-1}=\frac{1}{\beta^2\ln2}$, and $R_L=\doublehat{R}$.
\end{proof}
What one ideally want from the interval $[R_L,R_U]$ would be dependence of $R_U$ and $R_L$ on observations. It can be seen from Theorem~\ref{thm:int}, however, only $R_U$ depends on $\alpha$. This means that the interval $[R_L,R_U]$ may not be very meaningful. Even with $R_U$, upper bounding erfc$(\cdot)$ with a single exponential term could result in loose upper bound. To check validity of the interval $[R_L,R_U]$ in the actual channel model, we consider the following first order autoregressive channel model which is also called AR(1) model. 
\begin{itemize}
\item Simulation channel model
\begin{subequations}
\label{eq:simch}
\begin{eqnarray}
h_0&=&h_{temp,0},\\
h_i&=&\sqrt{1-C^2}h_{temp,i}+Ch_{i-1},\quad i=1,2,...
\end{eqnarray}
\end{subequations}
where $h_{temp,i} \sim \mathcal{CN}(0,1)$ and iid. 
\end{itemize}
The above model has the following properties.
\begin{subequations}
\begin{eqnarray}
E[h_ih^*_i]&=&1-C^{2i} \xrightarrow{i \rightarrow \infty} 1\\
E[h_ih^*_{i-1}]&=&C(1-C^{2i})\xrightarrow{i \rightarrow \infty} C.
\end{eqnarray}
\end{subequations}
Hence, the above channel model is the same as the one considered in Section~\ref{sec:al} with sufficiently large $i$ if we assume $l_d=1$ without loss of generality.\\
\indent
Table~\ref{tab:size} shows the ratio of the empirical average of $R_U-R_L$ to $\log_2(1+P)$ over $10^4$ realizations of $h_i$ with respect to various $P$ and $C$. 
\begin{table}[!t]\footnotesize
\renewcommand{\arraystretch}{1.3}
\caption{the ratio of the search interval $[R_L,R_U]$ to $\log_2(1+P)$}
\label{tab:size}
\centering
\begin{tabular}{|c||c|c|c|c|c|c|}
 \hline
   \backslashbox{Correlation}{SNR (dB)} & 0 (1)& 4 (1.81) &8 (2.87) &12 (4.07) &16 (5.35) &20 (6.66)  \\
  \hline
  \hline
  0.1 & 0.22 & 0.16 & 0.12 & 0.10 & 0.08 & 0.07\\
  \hline
  0.7 & 0.53 & 0.42 & 0.36 & 0.31 & 0.28 & 0.26\\
  \hline
  0.9 & 0.67 & 0.59 & 0.55 & 0.49 & 0.44 & 0.43\\
  \hline
  0.95 & 0.75 & 0.69 & 0.67 & 0.60 & 0.54 & 0.51\\
  \hline   
\end{tabular}
\end{table}
The value of $\log_2(1+P)$ which can be thought as the capacity of the average channel is given as reference in the bracket right next to each value of $P$. It can be seen from Table~\ref{tab:size} that the ratio of $R_U-R_L$ to $\log_2(1+P)$ decreases as $P$ increases, and it is more dramatic with low correlation. $R_U$ is determined by maximum of two terms as in Theorem~\ref{thm:int}, and which term being active affects the results given in Table~\ref{tab:size}. Table~\ref{tab:act} shows the empirical probability of $\hat{R}$ being active for $R_U$ over $10^4$ realizations of $h_i$. It can be seen that $\hat{R}$ is active mostly for low correlation. 
\begin{table}[!t]\footnotesize
\renewcommand{\arraystretch}{1.3}
\caption{Empirical probability of $R_U=\hat{R}$}
\label{tab:act}
\centering
\begin{tabular}{|c||c|c|c|c|c|c|}
 \hline
   \backslashbox{Correlation}{SNR (dB)} & 0& 4&8&12&16&20  \\
  \hline
  \hline
  0.1 & 1&1& 1&1&1&1\\
  \hline
  0.7 & 0.86&0.77&0.67&0.56&0.45&0.38\\
  \hline
  0.9&0.50& 0.42&0.35&0.27&0.22&0.18\\
  \hline
  0.95& 0.33&0.27&0.23&0.17&0.14&0.11\\
  \hline   
\end{tabular}
\end{table}
If $R_U=\hat{R}$, then $R'2^{R_U}-R_U'2^{R_L}=\frac{\alpha\sqrt{\frac{\pi}{2}}}{\beta^2\ln2}$, which implies that $R_U-R_L$ would increases as $P$ increases. The rate of this increase, however, would be slower than that of increase for $\log_2(1+P)$ due to the multiplicative factor $R_U$ and $R_L$ in $R'2^{R_U}-R_U'2^{R_L}=\frac{\alpha\sqrt{\frac{\pi}{2}}}{\beta^2\ln2}$, and it turns out to be much slower as can be seen in Table~\ref{tab:size} for low correlation case.\\
 \indent
High correlation cases show slightly different trends, and it may imply looseness of the interval $[R_L,R_U]$. Intuitively, it would be more likely that the realization of the next channel instance is close to the observation when correlation is high. In that case, search interval must be smaller than low correlation cases, but results in Table~\ref{tab:size} actually show the opposite. Looseness of $R_U$ is possibly a reason for this phenomenon. Note that the upper bound of Marcum Q-function which is used in the proof of Theorem~\ref{thm:int} is valid only for $\beta\gamma(R)>\alpha$. Because of this, $R_U$ is expressed as maximum of two terms, and $R_U$ is likely loose if $R_U=\log_2 (\frac{\alpha^2}{\beta^2}+1 )$. Table~\ref{tab:act} shows that it is more likely that $R_U=\log_2 (\frac{\alpha^2}{\beta^2}+1 )$ when correlation is high, and this can be another reason why the interval $[R_L,R_U]$ does not become small with high correlation. Such possible looseness of the interval $[R_L,R_U]$ can be troublesome especially given the fact that the proposed scheme would likely be meaningful only for sufficiently high correlation, and this means that further optimization of the interval $[R_L,R_U]$ may be needed. Nevertheless, absolute size of the interval $[R_L,R_U]$ is not impractically large according to Table~\ref{tab:size}, and hence, performance evaluation of the proposed scheme in this paper will be based on the search interval in Theorem~\ref{thm:int}.
%%%%%%%%%%%%%%%%%%%%%%%%%%%%%%%%%%%%%%%%%%%%%%
%%%%%%%%%%%%%%%%%%%%%%%%%%%%%%%%%%%%%%%%%%%%
\subsection{Performance evaluation}
\label{sec:perf}
To evaluate performance of the proposed scheme, we consider the following conventional schemes to compare with. 
\begin{enumerate}
 \item Scheme 1 (A scheme which ignores observation)\\
Since this scheme ignores observation, the transmission data rate of this scheme is fixed throughout whole transmission. This scheme treats each channel instance as a random variable which is circularly symmetric white Gaussian. Then, similar to the proposed scheme, the rate $\tilde{R}_1$ of this scheme is determined to maximize the expected throughput at each instance, and it is the solution of $\max_R \tilde{TP}(R)$, where $\tilde{TP}(R)=R \exp  ( -\frac{2^R-1}{P})$. 
\item Scheme 2 (A scheme which predicts the current channel based on observation)\\
It would be natural to consider a method which minimizes mean square error. The best prediction $\hat{h_i}$ of $h_i$ in that sense is $\hat{h_i}=E[h_i|h_{obs(i)}]$. If we consider the channel model in~\eqref{eq:simch}, then $\hat{h_i}=E[h_i|h_{i-1}]=Ch_{i-1}$. Then, the transmission data rate $\tilde{R}_{2,i}$ must be the capacity of the predicted channel which is $\tilde{R}_{2,i}=\log_2(1+C^2|h_{i-1}|^2)$.
 \end{enumerate}
It can be seen that $\tilde{R}_1$ would be the allocated rate of the proposed scheme when the channel correlation $C=0$, i.e., $\alpha=0$. In fact, the channel is ergodic when $C=0$, and hence, both the proposed scheme and Scheme 1 are optimal in terms of long-term throughput. This also means that performances of these two schemes would not be much different when channel correlation is not high enough. Furthermore, it can be seen from Theorem~\ref{thm:int} that $R_U=R_L$ when $\alpha=0$. Therefore, $\tilde{R}_1=R_U$ where $R_U$ is the solution of $R2^R=\frac{P}{\ln 2}$. \\
\indent
As described above, Scheme 1 and Scheme 2 have closed form solutions for the allocated rate, and hence, only the proposed scheme requires search of the allocated rate. We consider search in the interval $[R_L,R_U]$ with 100 evenly distributed points. Figure~\ref{fig:tp} shows the throughput results for various correlation values $C$. $10^4$ channel realizations are considered with the model in~\eqref{eq:simch} and the assumption of block error resulted from outage event is used for results in Figure~\ref{fig:tp}.
\begin{figure}[ht]
  \subfigure[$C=0.1$]{\label{fig:1}\includegraphics[width=0.5\textwidth]{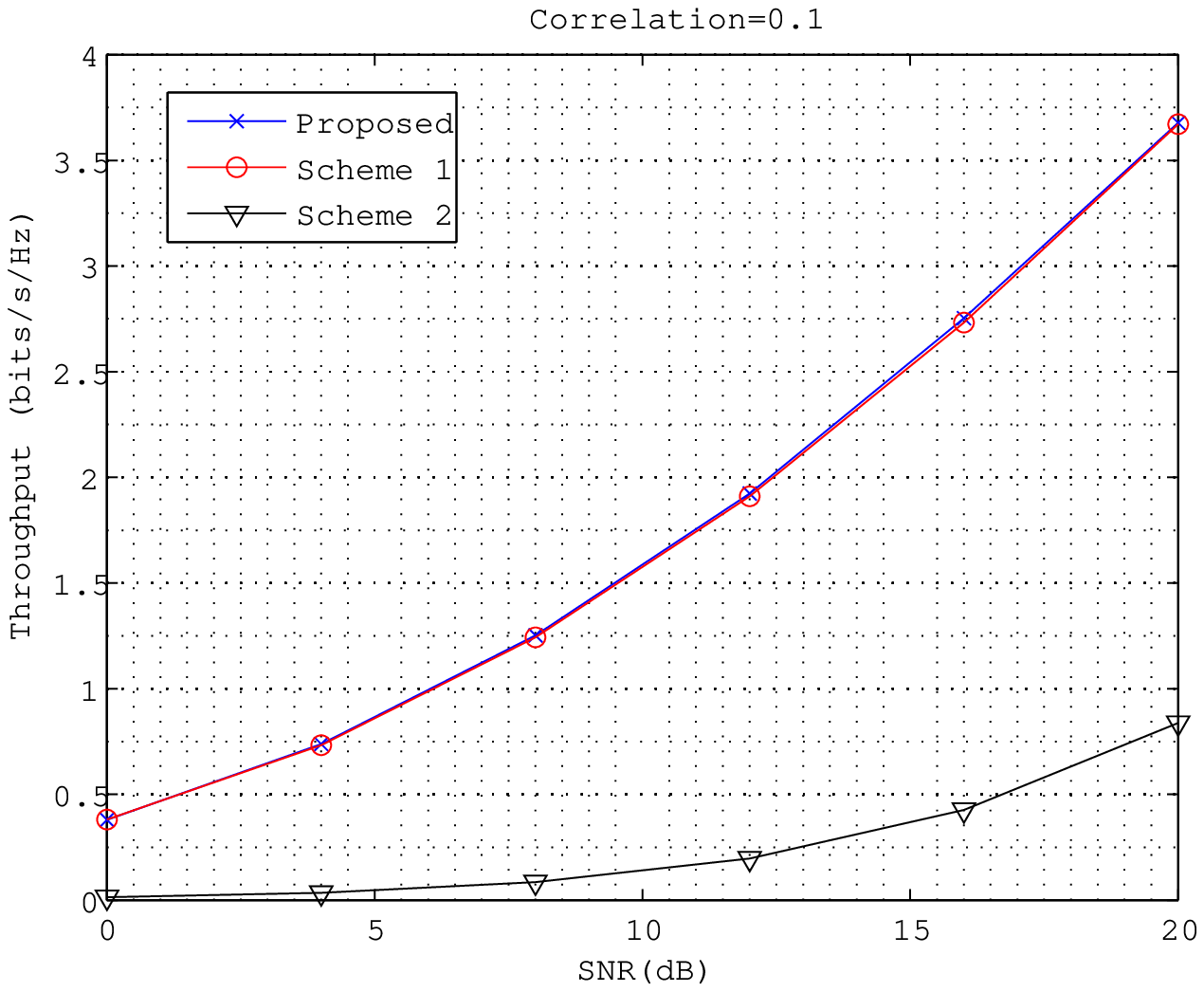}} 
  \subfigure[$C=0.7$]{\label{fig:7}\includegraphics[width=0.5\textwidth]{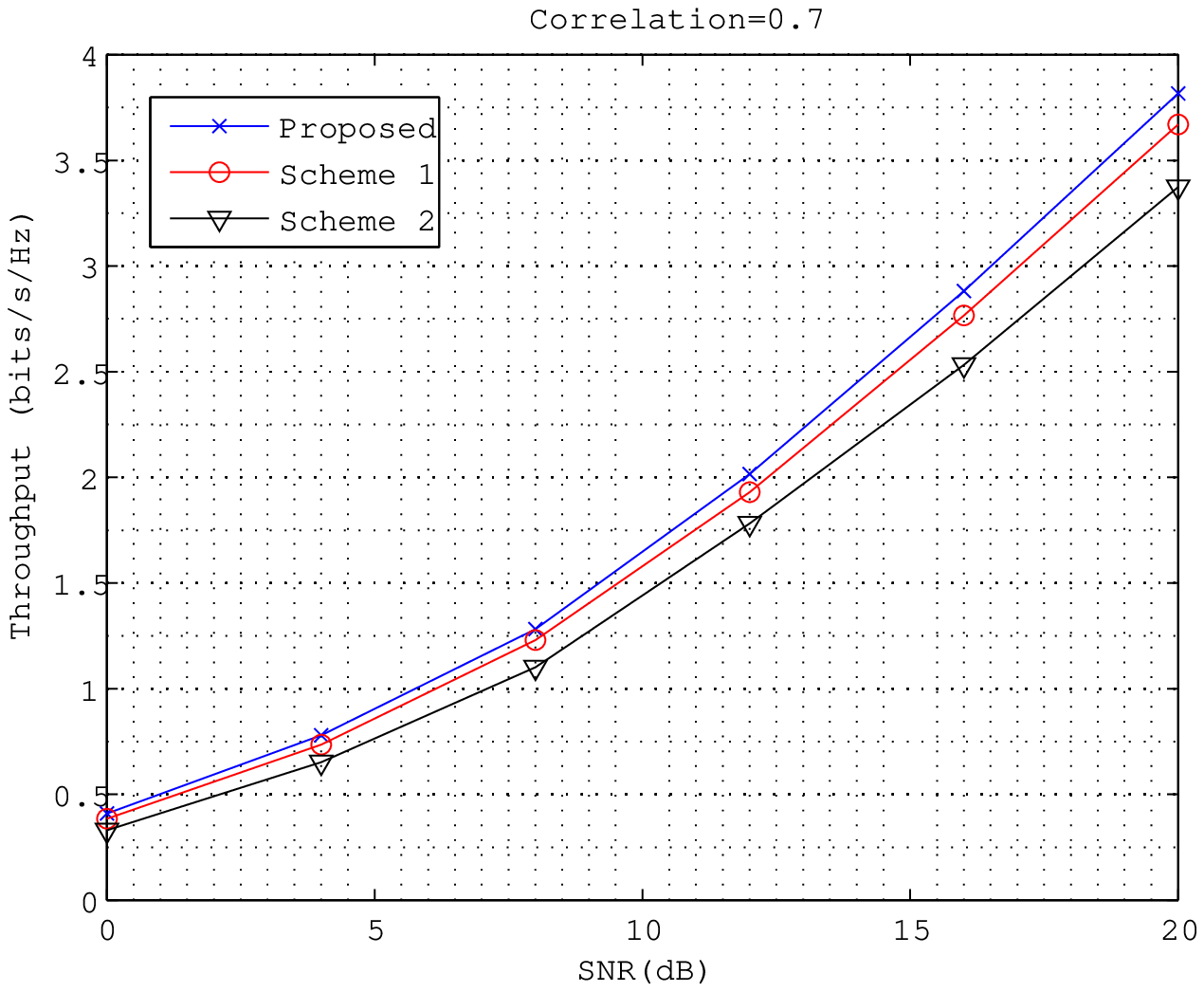}}
  \subfigure[$C=0.9$]{\label{fig:9}\includegraphics[width=0.5\textwidth]{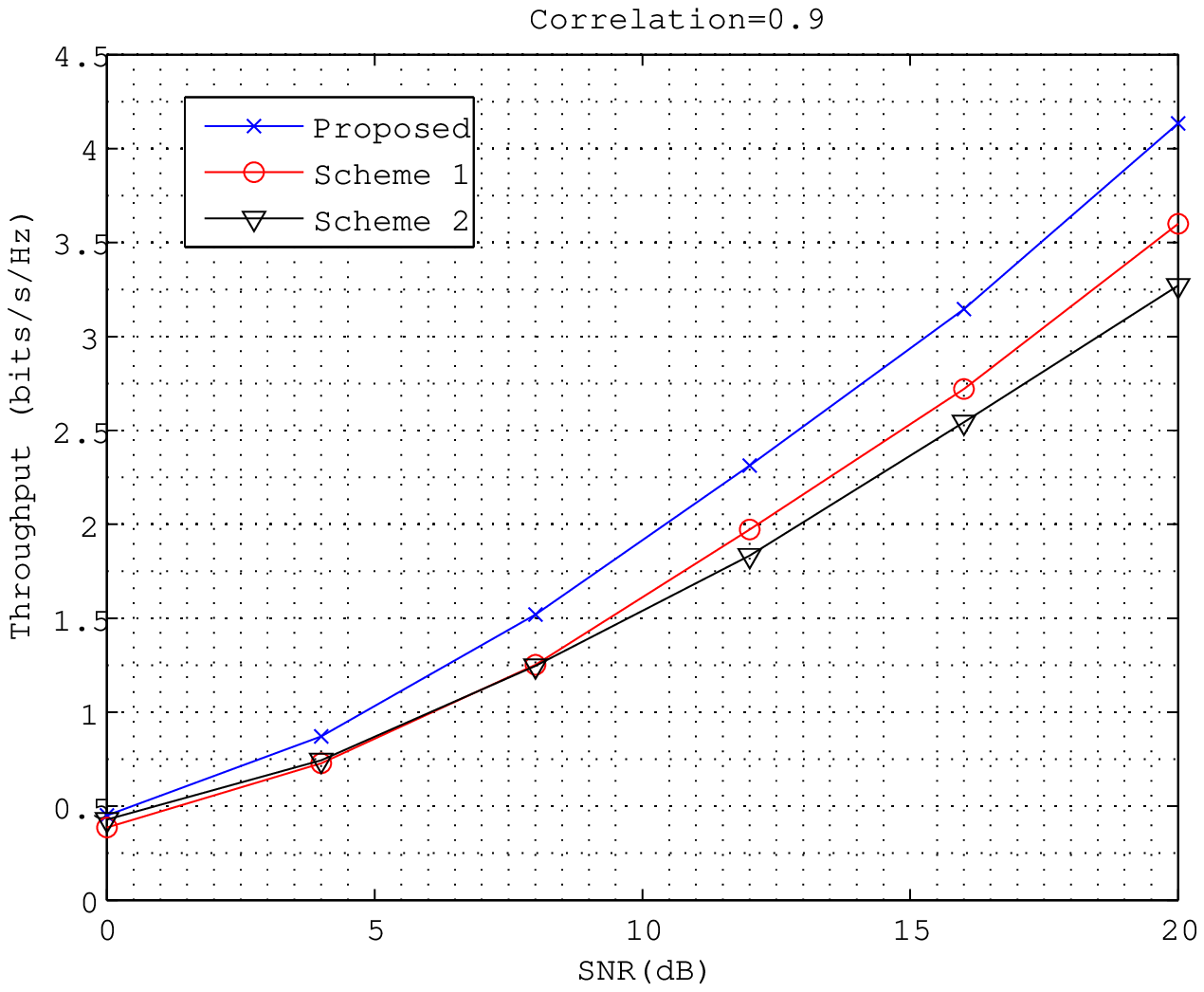}}
  \subfigure[$C=0.95$]{\label{fig:95}\includegraphics[width=0.5\textwidth]{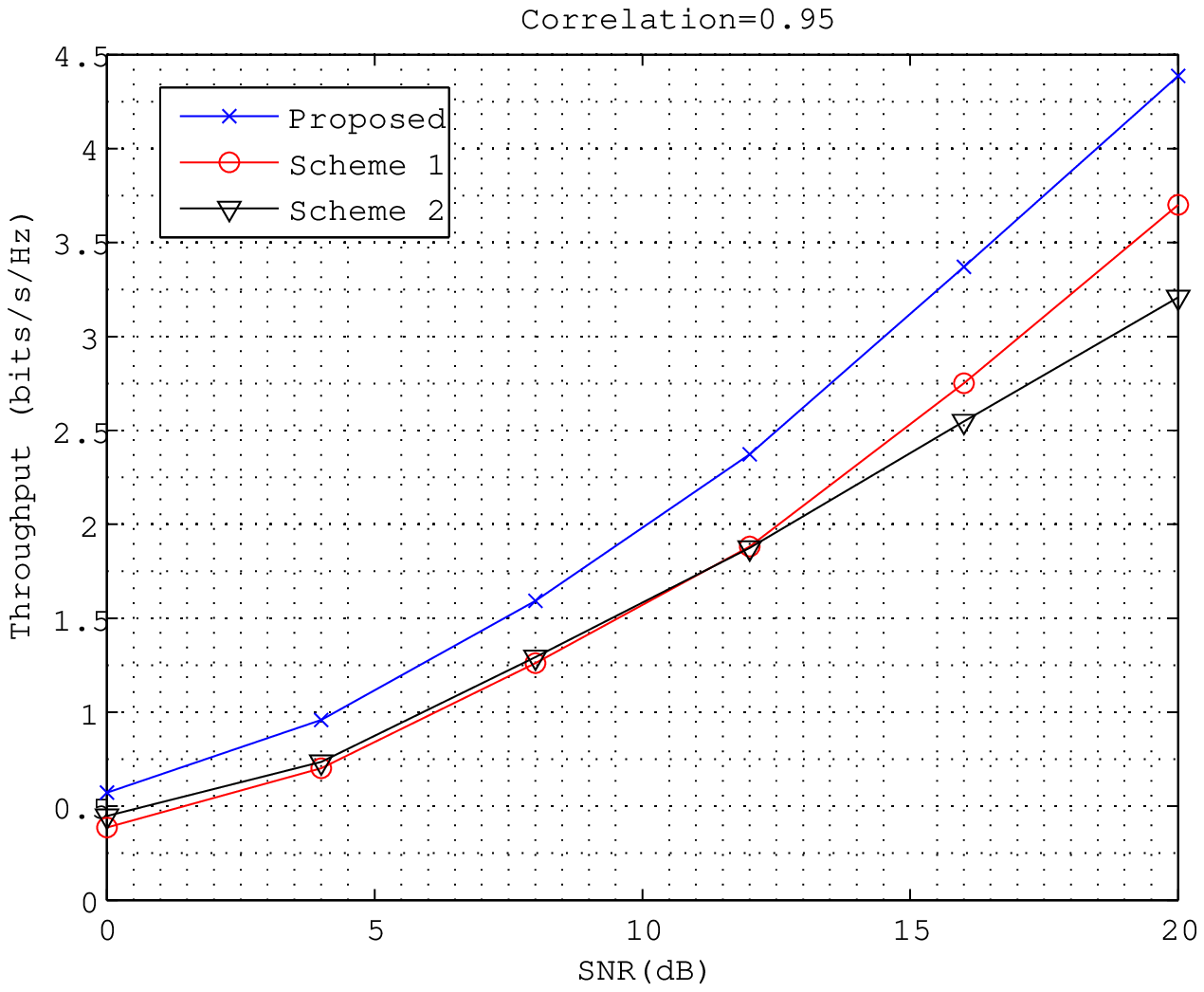}}
  \caption{Evaluated throughput for various correlations}
  \label{fig:tp}
\end{figure}
It can be seen that the proposed scheme shows almost identical performance to Scheme 1 which ignores observation with low correlation as expected. As correlation increases, the performance gap between these two gets larger. Interestingly, Scheme 2 performs very poorly with low correlation, and it does not outperform Scheme 1 even with $C=0.95$. This is probably because of the restriction of one observation, and it can be different with more observations. It is clear, however, that the proposed scheme performs significantly better than a scheme which tries to predict current channel. Even with more observations, the proposed scheme should benefit from it as well, and hence, performance gap would likely remain. 
%%%%%%%%%%%%%%%%%%%%%%%%%%%%%%%%%%%%%%%%%%%%%%%%%%
\subsection{Evaluation of ergodicity}
\label{sec:erg}
Although the proposed scheme outperforms the conventional ones as described in Figure~\ref{fig:tp}, channel correlation needs to be sufficiently high for that to happen. If the proposed scheme does not achieve the optimal performance, then there may be a room for improvement. As mentioned earlier, the proposed scheme is optimal if the channel has ergodic property. Consider the code block $i$'s which have the same value of $|h_{i-l_d}|$'s. Optimality of the proposed scheme is determined by whether the actual throughput of such code blocks converges to $\max_R TP_i(R)$. Let $I_A(\cdot)$ is an indicator function. In other words, $I_A(x)=1$ if $x \in A$ and $I_A(x)=0$ otherwise. We can think of $RI_{ \{h_i: \log_2 (1+P|h_i|^2)\geq R \} }(h_i)$ as a random variable for $i$'s with $|h_{i-l_d}|=x$ where $x\geq 0$. Then, the statistical average of this random variable would be $TP_i(R)$ given $|h_{i-l_d}|=x$. We want this statistical average to be equal to the limit of empirical average given as $\lim_{T \rightarrow \infty}\frac{R\sum_{i=l_d+1}^T I_{ \{h_i: \log_2 (1+P|h_i|^2)\geq R \} }(h_i)I_{\{h_{i-l_d}:|h_{i-l_d}|=x \}}(h_{i-l_d})}{\sum_{i=l_d+1}^T I_{\{h_{i-l_d}:|h_{i-l_d}|=x \}}(h_{i-l_d})}$. The following is the formal definition of this ergodicity. 
\begin{definition}
A Rayleigh fading process $\{h_i\}_{i=1}^T$ where $h_i$'s are jointly Gaussian is called weakly TP-ergodic if it satisfies the following for $R, x\geq 0$ with probability 1.
\begin{equation}
TP_i(R)=\lim_{T \rightarrow \infty}\frac{R\sum_{i=l_d+1}^T I_{ \{h_i: \log_2 (1+P|h_i|^2)\geq R \} }(h_i)I_{\{h_{i-l_d}:|h_{i-l_d}|=x \}}(h_{i-l_d})}{\sum_{i=l_d+1}^T I_{\{h_{i-l_d}:|h_{i-l_d}|=x \}}(h_{i-l_d})},
\end{equation}
where $TP_i(R)$ is computed for $|h_{i-l_d}|=x$.
\end{definition}
We can also think of pseudo ergodicity which is closely related to weak TP-ergodicity and has better operational implication.
\begin{definition}
A Rayleigh fading process $\{h_i\}_{i=1}^T$ where $h_i$'s are jointly Gaussian is called pseudo weakly TP-ergodic if it satisfies the following for $R\geq 0$ with probability 1.
\begin{equation}
\lim_{T \rightarrow \infty} \frac{\sum_{i=l_d+1}^T TP_{i}(R)}{T}=\lim_{T \rightarrow \infty}\frac{R}{T}\sum_{i=l_d+1}^T I_{ \{h_i: \log_2 (1+P|h_i|^2)\geq R \} }(h_i).
\end{equation}
\end{definition}
In fact, the proposed scheme is optimal as long as it satisfies the above pseudo weak TP-ergodicity. Let us now look at the relationship between two ergodicity conditions. 
\begin{lemma}
Weak TP-ergodicity implies pseudo weak TP-ergodicity if 
\begin{eqnarray}
\label{eq:lim}
&&\lim_{T \rightarrow \infty}\frac{R\sum_{i=l_d+1}^T I_{ \{h_i: \log_2 (1+P|h_i|^2)\geq R \} }(h_i)I_{\{h_{i-l_d}:|h_{i-l_d}|=x \}}(h_{i-l_d})}{\sum_{i=l_d+1}^T I_{\{h_{i-l_d}:|h_{i-l_d}|=x \}}(h_{i-l_d})}\nonumber\\
&&\quad = \frac{\lim_{T \rightarrow \infty}\frac{R}{T}\sum_{i=l_d+1}^T I_{ \{h_i: \log_2 (1+P|h_i|^2)\geq R \} }(h_i)I_{\{h_{i-l_d}:|h_{i-l_d}|=x \}}(h_{i-l_d})}{\lim_{T \rightarrow \infty} \frac{1}{T}\sum_{i=l_d+1}^T I_{\{h_{i-l_d}:|h_{i-l_d}|=x \}}(h_{i-l_d})}.
\end{eqnarray}
\end{lemma}
\begin{proof}
Assume that weak TP-ergodicity holds. If the assumption~\eqref{eq:lim} holds, then weak TP-ergodicity can be written as
\begin{eqnarray}
&&\lim_{T \rightarrow \infty}\frac{1}{T}\sum_{i=l_d+1}^T TP_i(R)I_{\{h_{i-l_d}:|h_{i-l_d}|=x \}}(h_{i-l_d})\nonumber\\
&&\quad =\lim_{T \rightarrow \infty}\frac{R}{T}\sum_{i=l_d+1}^T I_{ \{h_i: \log_2 (1+P|h_i|^2)\geq R \} }(h_i)I_{\{h_{i-l_d}:|h_{i-l_d}|=x \}}(h_{i-l_d}).
\end{eqnarray}
Consequently, we get 
\begin{eqnarray}
&&\int_x \lim_{T \rightarrow \infty}\frac{1}{T}\sum_{i=l_d+1}^T TP_i(R)I_{\{h_{i-l_d}:|h_{i-l_d}|=x \}}(h_{i-l_d}) dx\nonumber\\
&&\quad =\int_x \lim_{T \rightarrow \infty}\frac{R}{T}\sum_{i=l_d+1}^T I_{ \{h_i: \log_2 (1+P|h_i|^2)\geq R \} }(h_i)I_{\{h_{i-l_d}:|h_{i-l_d}|=x \}}(h_{i-l_d})dx.
\end{eqnarray}
Since $\frac{1}{T}\sum_{i=l_d+1}^T TP_i(R)I_{\{h_{i-l_d}:|h_{i-l_d}|=x \}}(h_{i-l_d})$ and $\frac{R}{T}\sum_{i=l_d+1}^T I_{ \{h_i: \log_2 (1+P|h_i|^2)\geq R \} }(h_i)I_{\{h_{i-l_d}:|h_{i-l_d}|=x \}}(h_{i-l_d})$ are upper bounded by $R$, limit and integration can be exchanged from dominated convergence theorem~\cite{Fo99}. Then, we have
\begin{equation}
\lim_{T \rightarrow \infty} \frac{\sum_{i=l_d+1}^T TP_{i}(R)}{T}=\lim_{T \rightarrow \infty}\frac{R}{T}\sum_{i=l_d+1}^T I_{ \{h_i: \log_2 (1+P|h_i|^2)\geq R \} }(h_i),
\end{equation}
which proves the claim.
\end{proof}
Pseudo weak TP-ergodicity can be checked numerically by comparing the empirical average of $\max_R TP_i(R)$ (here, the observation $|h_{i-l_d}|$ is random) with the actual throughput. We can also think of stronger ergodicity as follows. 
\begin{definition}
A Rayleigh fading process $\{h_i\}_{i=1}^T$ where $h_i$'s are jointly Gaussian is called strongly TP-ergodic if it satisfies the following for all $R$ with probability 1.
\begin{equation}
E_{h_{i-l_d}}\Big[TP_{i}(R)\Big]=\lim_{T \rightarrow \infty}\frac{R}{T}\sum_{i=l_d+1}^T I_{ \{h_i: \log_2 (1+P|h_i|^2)\geq R \} }(h_i).
\end{equation}
\end{definition}
The strong TP-ergodicity says that the actual throughput converges to statistical average of $\max_R TP_i(R)$ with respect to marginal distribution of $h_{i}$. Strong TP-ergodicity requires $E_{h_{i-l_d}}[TP_{i}(R)]=\lim_{T \rightarrow \infty} \frac{\sum_{i=l_d+1}^T TP_{i}(R)}{T}$ in addition to pseudo weak TP-ergodicity. Note that this strong TP-ergodicity does not need to hold for the proposed scheme to be optimal. Figure~\ref{fig:erg} describes the actual throughput, the empirical average of $\max_R TP_i(R)$, and the statistical average of $\max_R TP_i(R)$ with various correlations using $10^4$ channel realizations with the model in~\eqref{eq:simch}. 
\begin{figure}[ht]
  \subfigure[$C=0.1$]{\label{fig:1erg}\includegraphics[width=0.5\textwidth]{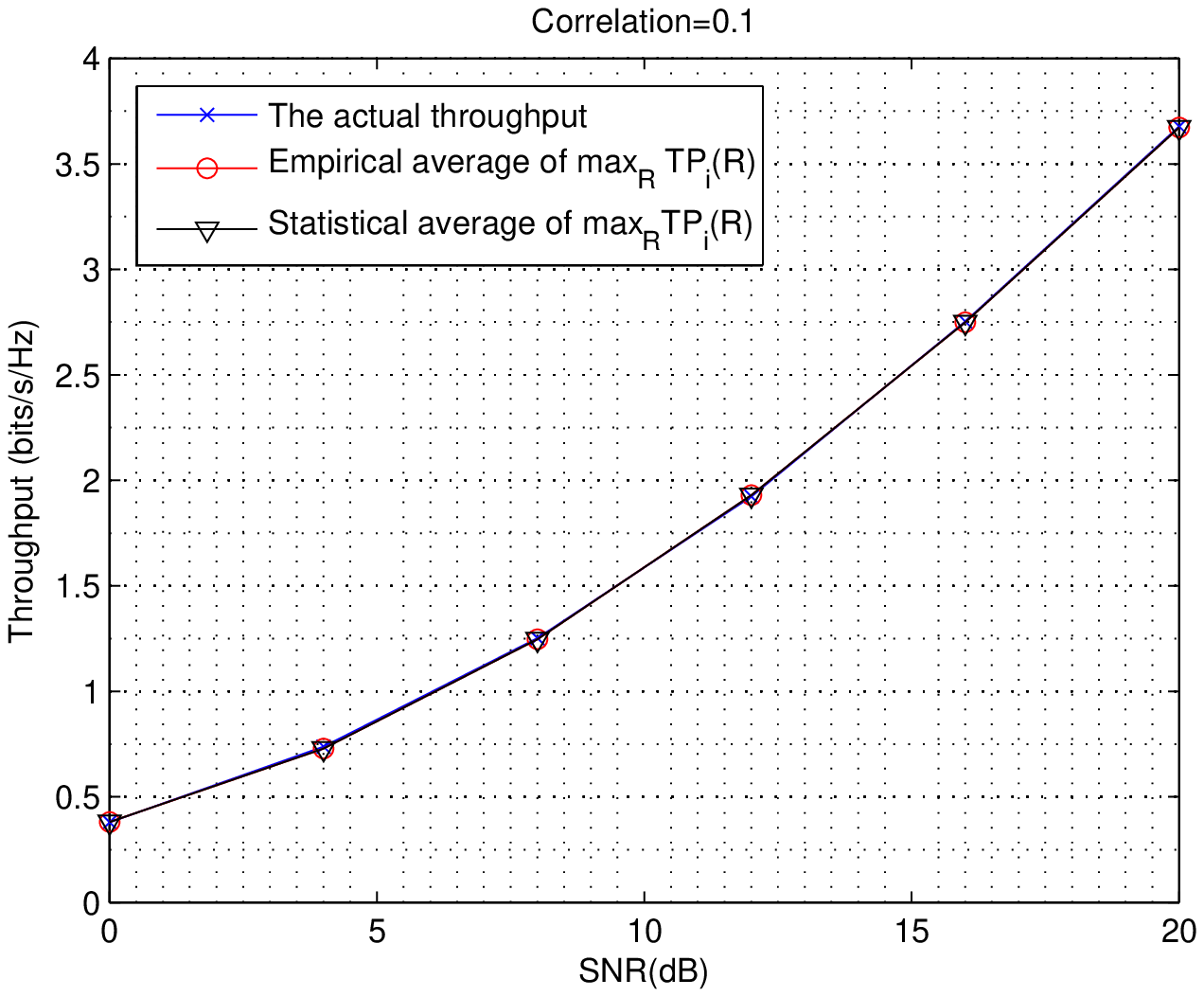}} 
  \subfigure[$C=0.7$]{\label{fig:7erg}\includegraphics[width=0.5\textwidth]{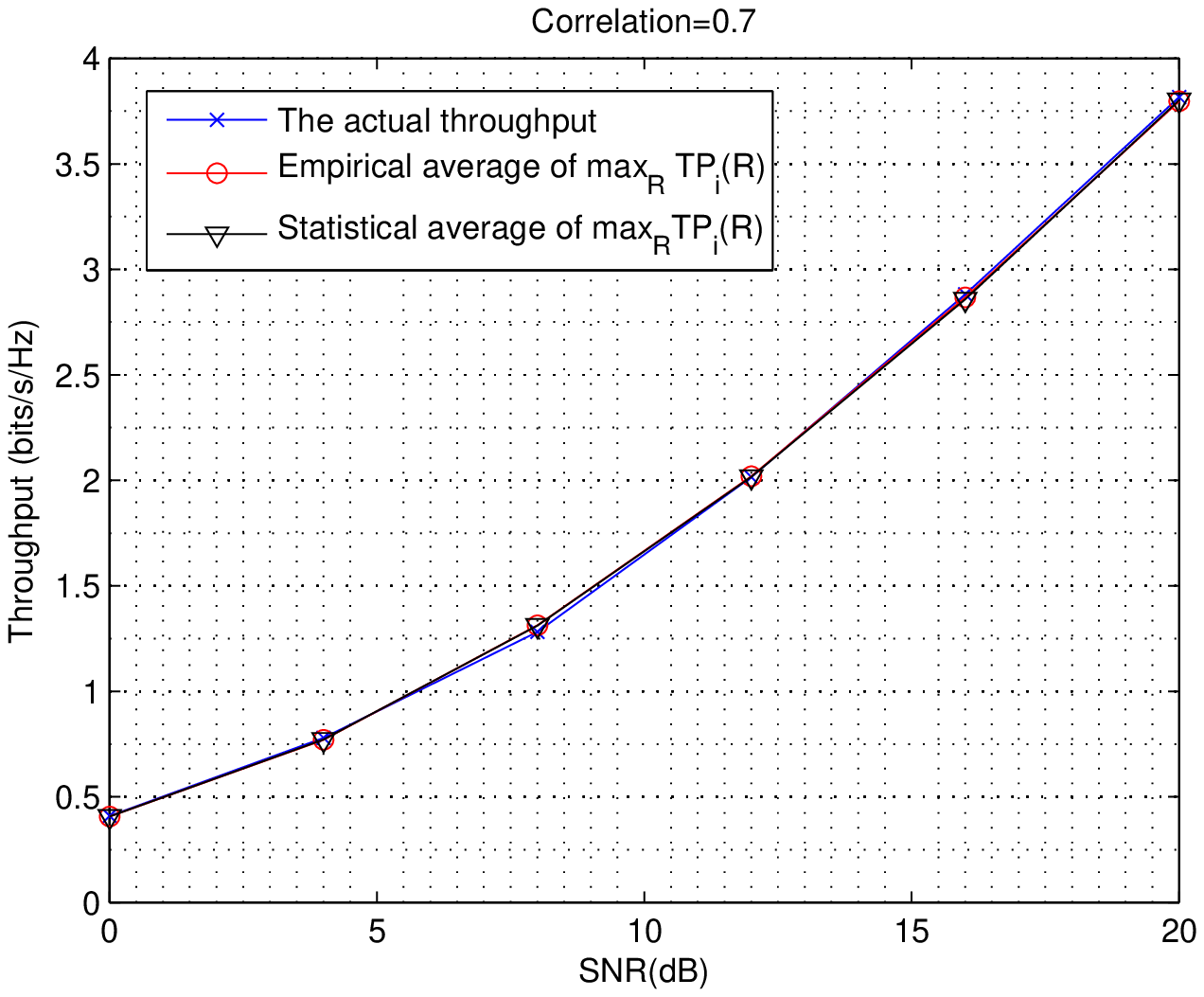}}
  \subfigure[$C=0.9$]{\label{fig:9erg}\includegraphics[width=0.5\textwidth]{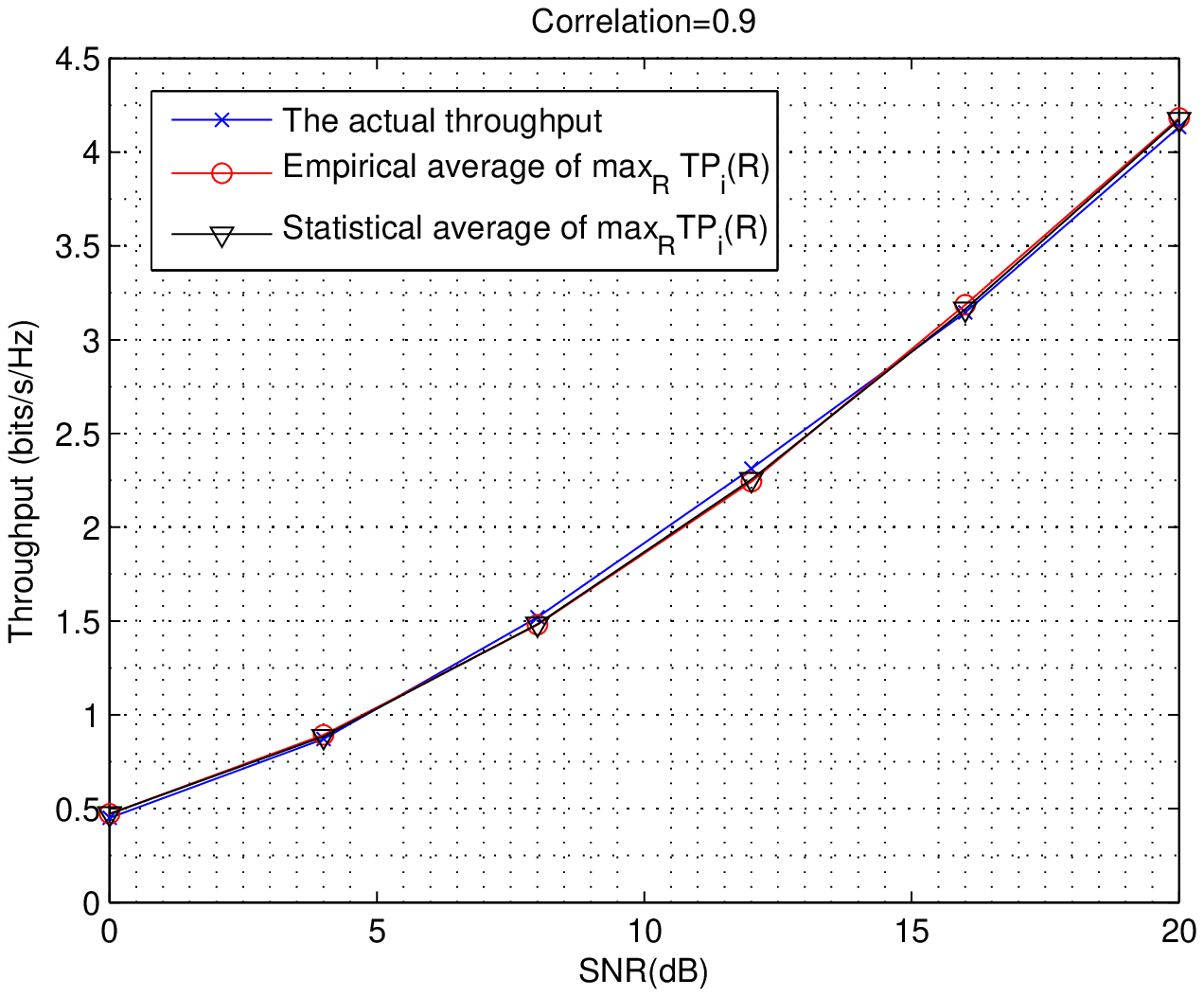}}
  \subfigure[$C=0.95$]{\label{fig:95erg}\includegraphics[width=0.5\textwidth]{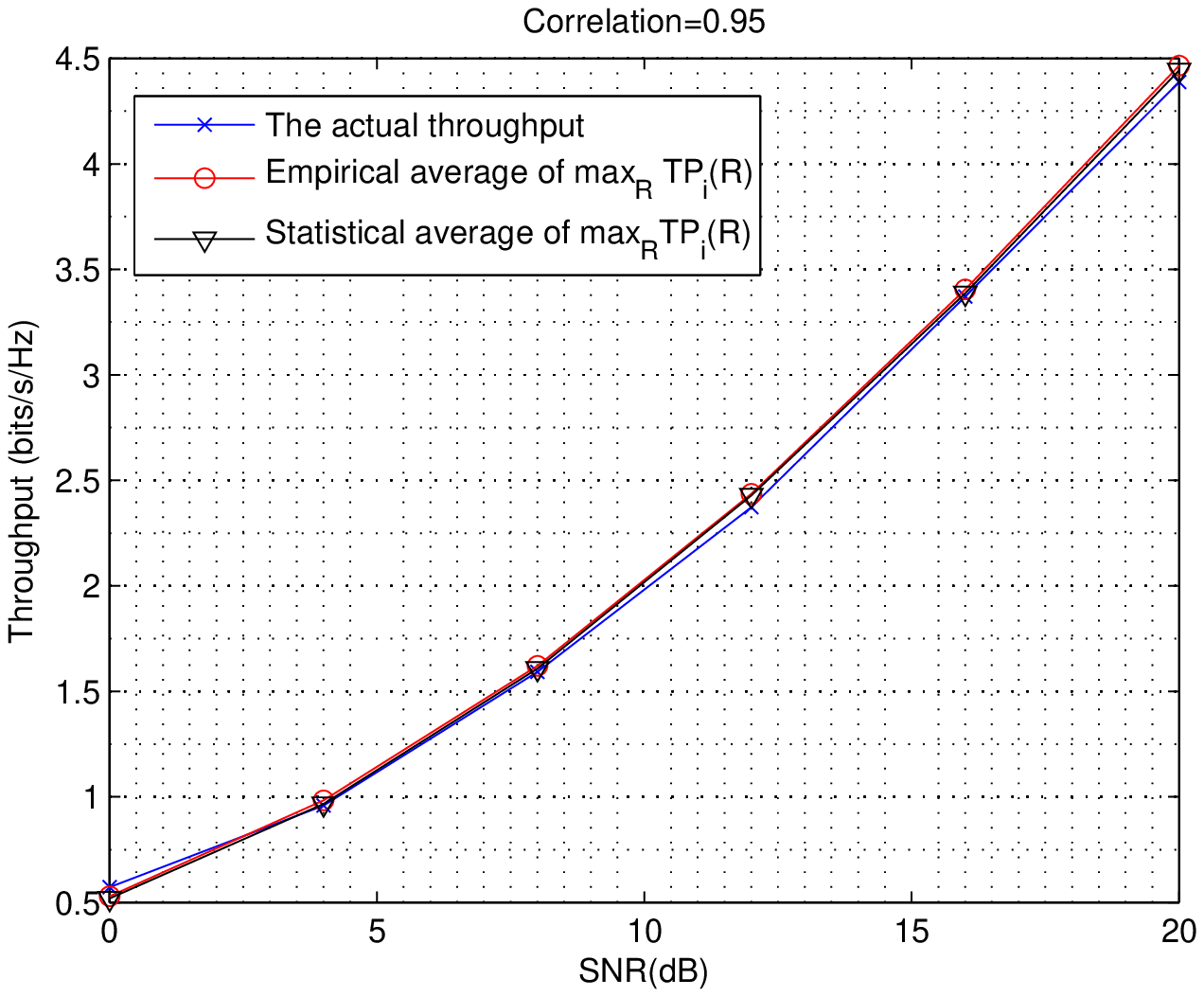}}
  \caption{Evaluation of ergodicity for various correlations}
  \label{fig:erg}
\end{figure}
It can be seen that they are very close for all channel correlations, which hints that the both strong and pseudo weak TP-ergodicity conditions hold, and that the proposed scheme is near optimal in terms of long-term throughput. 
%%%%%%%%%%%%%%%%%%%%%%%%%%%%%%%%%%%%%%%%%%%%%%%%%
%%%%%%%%%%%%%%%%%%%%%%%%%%%%%%%%%%%%%%%%%%%%%%%%%%
\section{Generalization of the proposed scheme}
\label{sec:ext}
The proposed scheme considered in Section~\ref{sec:out} is defined only with SISO block flat fading model given in~\eqref{eq:ch}. As mentioned earlier, it can be generalized to more complicated cases, and this section discusses such generalization. 
\subsection{SIMO}
\label{sec:simo}
Consider the following single-input multi-output (SIMO) block flat fading model. 
\begin{equation}
\underline{y}_i[n] = \sqrt{P}\underline{h}_ix_i[n] + \underline{z}_i[n] \text{ for } n = 1,...,N,
\end{equation}
where $\underline{h}_i$ is a vector of length $M$ where $M$ is the number of receive antennas. There are $M$ elements in $\underline{h}_i$, and they are independent with each other. Also, $M$ elements of $\underline{z}_i[n]$ are independent with each other. Let $h_{i,m}$ be the $m$th element of $\underline{h}_i$, and $z_{i,m}[n]$ be the $m$th element of $\underline{z}_i[n]$. Similar to the SISO case, we assume that background noise $z_{i,m}[n] \sim \mathcal{CN}(0,1)$ is iid. $\{h_{i,m}\}_{i=1}^\infty$ is Rayleigh fading process which means that $h_{i,m}\sim  \mathcal{CN}(0,1)$, and $P$ represents SNR. We assume that $\{h_{i,m}\}_{i=1}^\infty$ is jointly Gaussian. The proposed scheme can directly be extended to this case. 
The capacity of the $i$th code block $C_{i}$ can be expressed as
\begin{equation}
C_i= \log_2 (1+P\|\underline{h}_{i}\|^2 ).
\end{equation}
With restriction of single observation, let $C=E[h_{i,m}h^*_{i-l_d,m}]$. Given $h_{i-l_d,m}$,
\begin{equation}
h_{i,m} \sim \mathcal{CN} (Ch_{i-l_d,m}, 1-C^2).
\end{equation}
The outage probability $P_{out}(R)$ of the $i$th code block for given $R$ can be expressed as 
\begin{subequations}
\begin{eqnarray}
P_{i,out}(R)&=& P \{ C_{i}<R | \underline{h}_{i-l_d} \}\\
          &=& P\Big\{\|h_{i}\|^2<\frac{2^R-1}{P}\Big| \underline{h}_{i-l_d}\Big\}.
\end{eqnarray}
\end{subequations}
It can be seen that $\|\underline{h}_{i}\|^2$ given $\underline{h}_{i-l_d}$ has noncentral chi-square distribution with $2M$ degrees of freedom, whose cdf is given as 
\begin{equation}
P\{|\underline{h}_{i}\|^2<x|\underline{h}_{i-l_d} \}=1-Q_{M}\bigg(  \sqrt{\frac{2C\|\underline{h}_{i-l_d}\|^2}{1-C^2}}, \sqrt{\frac{2x}{1-C^2}}\bigg),
\end{equation}
where $Q_{M}$ is a Marcum Q-function defined as
\begin{equation}
Q_M(a,b)=\int_b^\infty \frac{x^{M}}{a^{M-1}}\exp\Big(-\frac{x^2+a^2}{2}\Big)I_{M-1}(ax)dx,
\end{equation}
where $I_{M-1}$ is a modified Bessel function of order $M-1$. Therefore, the expected throughput $TP_i(R)$ of the $i$th code block for given $R$ as
\begin{subequations}
\begin{eqnarray}
TP_i(R)&=&R(1-P_{i,out}(R))\\
       &=&RQ_M\bigg(  \sqrt{\frac{2C\|\underline{h}_{i-l_d}\|^2}{1-C^2}}, \sqrt{\frac{2(2^R-1)}{P(1-C^2)}}\bigg).
\end{eqnarray}
\end{subequations}
Then, the proposed scheme is to choose $R$ which is the solution of $\max_R TP_i(R)$ as the transmission data rate of the $i$th code block. Similar to SISO case, the closed form solution cannot be found, which means that determination of search interval as in Theorem~\ref{thm:int} must be done to implement the scheme. As for $Q_1(\cdot)$, there are several results which consider bounds on general Marcum Q-function $Q_M(\cdot)$, and they may be used to determine such interval. 
%%%%%%%%%%%%%%%%%%%%%%%%%%%%%%%%%%%%%%%%%%%%%%%%%%%%%%%%%%
%%%%%%%%%%%%%%%%%%%%%%%%%%%%%%%%%%%%%%%%%%%%%%%%%%%%%%
\subsection{Fading within a code block}
\label{sec:fad1}
Consider now the following SISO flat fading model. 
\begin{equation}
y_i[n] = \sqrt{P}h_i[n]x_i[n] + z_i[n] \text{ for } n = 1,...,N,
\end{equation}
We assume that background noise $z_i[n] \sim \mathcal{CN}(0,1)$ is iid. $\{h_i[1],...,h_i[N]\}_{i=1}^\infty$ is Rayleigh fading process which means that $h_i[n]\sim  \mathcal{CN}(0,1)$, and $P$ represents SNR. We assume that $\{h_i[1],...,h_i[N]\}_{i=1}^\infty$ is jointly Gaussian. The capacity $C_i$ of the $i$th code block can be calculated as 
\begin{equation}
C_i=\max_{P(X^N)}\lim_{N \rightarrow \infty}\frac{1}{N}I(X^N;Y^N).
\end{equation}
Note that
\begin{subequations}
\begin{eqnarray}
I(X^N;Y^N)&=&h(Y^N)-h(Y^N|X^N)\\
          &\stackrel{(a)}{=}&h(Y^N)-\sum_{n=1}^N h(Y_n|X_n)\\
          &\leq& \sum_{n=1}^N h(Y_n)-\sum_{n=1}^N h(Y_n|X_n)\\
          &=& \sum_{n=1}^N I(X_n;Y_n),
\end{eqnarray}
\end{subequations}
where (a) comes from the fact that background noise is iid. Therefore, 
\begin{equation}
C_i\leq \lim_{N \rightarrow \infty}\frac{1}{N} \sum_{n=1}^N  \log_2 (1+P|h_{i}[n]|^2 ).
\end{equation}
The equality of the above expression is achieved when $y[n]$ is independent, and this depends on the joint statistics of the channel. Since it would be difficult to derive the exact capacity $C_i$ in this case, we consider the upper bound $\frac{1}{N} \sum_{n=1}^N  \log_2 (1+P|h_{i}[n]|^2 )$ as the approximation of the capacity for finite $N$. The conditional probability of $h_i[1],...,h_i[N]$ given observation can be expressed by using Lemma~\ref{lem:gau} as in Section~\ref{sec:al} given that correlation among all involved random variables is specified. The expression of the outage probability and the expected throughput, however, would be much more complicated than those in Seciton~\ref{sec:al} or Section~\ref{sec:simo}. Hence, we consider furhter simplified expression as follows. 
\begin{equation}
\label{eq:simpc}
C_i \approx \log_2 \Big(1+\frac{P}{N}\sum_{n=1}^N|h_{i}[n]|^2\Big ).
\end{equation}        
Note that we are further upper bounding the capacity by doing this because of concavity of log function, which implies that using this will underestimates outage probability. Let us assume again a single code block of observation. Let $(h_i[1],...,h_i[N])^T$ be $\underline{h}_i$. We first interested in the distribution of $\underline{h}_{i}$ given $\underline{h}_{i-l_d}$. By using Lemma~\ref{lem:gau},
\begin{equation}
\underline{h}_{i} \sim \mathcal{CN} (\tilde{A}\underline{h}_{i-l_d}, \tilde{C}_3-\tilde{A}\tilde{C}_1^H),
\end{equation}
where $\tilde{A}=\tilde{C}_1\tilde{C}_2^{-1}$, $\tilde{C}_1=E[\underline{h}_{i}\underline{h}_{i-l_d}^H]$, $\tilde{C}_2=E[\underline{h}_{i-l_d}\underline{h}_{i-l_d}^H]$, and $\tilde{C}_3=E[\underline{h}_{i}\underline{h}_{i}^H]$. Since $\tilde{C}_3-\tilde{A}\tilde{C}_1^H$ may not be a diagonal matrix, $\sum_{n=1}^{N}|h_{i}[n]|^2$ given $\underline{h}_{i-l_d}$ has generalized chi-squared distribution which has no known closed functional form. If we assume that $\tilde{C}_3-\tilde{A}\tilde{C}_1^H$ is diagonal, then $\sum_{n=1}^{N}|h_{i}[n]|^2$ given $\underline{h}_{i-l_d}$ has noncentral chi-squared distribution with $2N$ degrees of freedom, and the resulting outage probability and the expected throughput can be derived as in Section~\ref{sec:simo}.
%%%%%%%%%%%%%%%%%%%%%%%%%%%
%%%%%%%%%%%%%%%%%%%%%%%%%%%%%%%%%%%
\subsection{HARQ}
The current wireless standard supports hybrid automatic repeat request (HARQ) to improve reliability of transmission. Under HARQ, the received signal at each transmission is not discarded even with incorrect decoding, and the retransmitted signal is combined with stored previous signals. There are two prevalent combining methods called Chase combining (CC) and incremental redundancy (IR). Detailed description of them can be found in~\cite{Ch85} for CC and in~\cite{LiCoMi84} for IR.\\
\indent
Let us consider a single retransmission for simplicity. We also consider SISO block flat fading model given in~\eqref{eq:ch} with a single code block observation. The transmitter determines the transmission data rate of code block $i$ with $h_{i-l_d}$. Assume that possible retransmission occurs at $(i+l_r)$th code block. For the both combining schemes, the outage event of the first transmission is the event in which the transmission data rate $R$ is greater than the channel capacity $C_i$. For IR, the outage event of the second transmission is the event in which the transmission data rate $R$ is greater than the sum of capacities of two channels $C_i+C_{i+l_r}$. From now on, we will restrict ourselves to IR. \\
\indent
Let $C=E[h_{i}h^*_{i-l_d}]$, $C'=E[h_{i}h^*_{i+l_r}]$ and $C''=E[h_{i-l_d}h^*_{i+l_r}]$. Let $P^1_{i,out}$ be outage probability of the first transmission. Then, it is given as before,
\begin{equation}
P^1_{i,out}(R)=1-Q_1\bigg(  \sqrt{\frac{2C|h_{i-l_d}|^2}{1-C^2}}, \sqrt{\frac{2(2^R-1)}{P(1-C^2)}}\bigg).
\end{equation}
Let $P^2_{i,out}$ be outage probability of the second transmission. To compute it, we need to know the distribution of $|h_{i+l_r}|^2$ given $h_{i}$ and $h_{i-l_d}$. From Lemma~\ref{lem:gau}, given $h_{i}, h_{i-l_d}$,
\begin{equation}
h_{i+l_r} \sim \mathcal{CN} (A(h_{i}, h_{i-l_d})^T, 1-AC_1^H),
\end{equation}
where $A=C_1C_2^{-1}$, $C_1=E[h_{i+l_r}(h^*_{i}, h^*_{i-l_d})]=(C', C'')$, and $C_2=E[(h_{i}, h_{i-l_d})^T(h^*_{i}, h^*_{i-l_d})]=\begin{bmatrix}
1 & C\\
C & 1
\end{bmatrix}$. 
Hence, $|h_{i+l_r}|^2$ given $h_{i}, h_{i-l_d}$ has noncentral chi-squared distribution with 2 degrees of freedom, and its cdf is given as 
\begin{equation}
\label{eq:cdf2}
P\{|h_{i+l_r}|^2<x|h_{i}, h_{i-l_d} \}=1-Q_1\bigg(  \sqrt{\frac{2|A(h_{i}, h_{i-l_d})^T|^2}{1-AC_1^H}}, \sqrt{\frac{2x}{1-AC_1^H}}\bigg).
\end{equation}
Then, $P^2_{i,out}$ is given as 
\begin{equation}
P^2_{i,out}(R)=E_{h_i|h_{i-l_d}}\bigg[ I(C_{i}<R) \bigg(1-Q_1\bigg(  \sqrt{\frac{2|A(h_{i}, h_{i-l_d})^T|^2}{1-AC_1^H}}, \sqrt{\frac{2(2^{R-C_{i}}-1)}{P(1-AC_1^H)}}\bigg)\bigg)\bigg ],
\end{equation}
where $I(\cdot)$ is indicator function. Without HARQ, the expected throughput can clearly be expressed as $TP_i(R)=R(1-P_{i,out}(R))$ as given in~\eqref{eq:tp}. With HARQ, however, there is no simple expression for the expected throughput. For example, let us consider the following expression as the expected throughput $TP_{i,HARQ}(R)$.
\begin{equation}
TP_{i,HARQ}(R)=R(1-P^1_{i,out}(R))+\frac{R}{2}(1-P^2_{i,out}(R)).
\end{equation}
The above expression is incorrect because it does not have consideration of the throughput at the $(i+l_r)$th code block with the first transmission being successful. Since the rate of the $(i+l_r)$th code block is not determined at the $i$th code block, i.e., it is the future decision, the rate allocation problem becomes the stochastic control problem with infinite horizon which is usually solved by dynamic programming.~\cite{Be05}. Solvability of such problem needs to be carefully investigated, and it will not be discussed here since its scope becomes out of this paper.\\
\indent
Lack of simple throughput expression means that the proposed scheme can only be sub-optimal if one sticks with some simple, approximate throughput expression. This problem, however, is not unique to the proposed scheme. Even if we consider an approach of predicting the channel, we would still have exactly the same problem. In fact, the approach of predicting channel has another problem under HARQ. If we are interested in rate allocation with consideration of possible retransmission, then we must be able to somehow characterize $P^1_{i,out}$. With the approach of predicting channel, the notion of outage probability is unclear, which means that it would be hard to design such a scheme which works with any kind of HARQ throughput expression depending on $P^1_{i,out}$. Note that the proposed scheme at least has a systematic way of determining this outage probability, and this could be significant advantage under HARQ. 
%%%%%%%%%%%%%%%%%%%%%%%%%%%%%%%%%%%%%%%%%%%%%%%%%%%%%%%%%%%%%%%%%%
%%%%%%%%%%%%%%%%%%%%%%%%%%%%%%%%%%%%%%%%%%%%%%%%%%%%%%%%%%%%%%%%%%%%%%%%%
%%%%%%%%%%%%%%%%%%%%%%%%%%%%%%%%%%%%%%%%%%%%%%%%%%%%%%%%%%%%%%%%%%%%%%%%
\section{Conclusion}
\label{sec:con}
In this paper, we have proposed an algorithm which maximizes long-term throughput for ergodic fading channel with delayed CSIT. For performance evaluation, we have considered AR(1) model with Rayleigh fading, and the proposed scheme considerably outperforms conventional ones in this model. Ergodicity conditions which need to be satisfied to ensure optimality of the proposed scheme are discussed, and numerical evaluation shows near optimality of the proposed scheme for AR(1) model. As long as channel has ergodicity, the proposed scheme must be optimal in terms of long-term throughput for any Rayleigh fading channel in which each channel symbol is correlated with finite number of channel symbols. Such ergodicity is generally hard to prove, but it is often assumed for various wireless channels as mentioned earlier. \\
\indent
The proposed scheme can also be generalized to other than SISO block fading model. In SIMO case, the expected throughput involves generalized Marcum Q-function $Q_M(\cdot)$ which possibly needs bounding analysis as given for $Q_1(\cdot)$ in this paper. The proposed scheme can be applied to the case of fading within a code block as well, although inexistence of a closed from functional expression for cdf of current channel would make things more complicated. For the case of HARQ, the correct throughput expression must be determined before applying the proposed scheme, but it still can be used if one is interested in maintaining certain BLER instead of maximizing the throughput.\\
\indent 
Although we have exclusively considered Rayleigh fading model in this paper, the proposed scheme can be applied to any channel model given that joint channel statistics are known. It can be computationally difficult, however, to maximize the expected throughput over transmission rate $R$ for general channel model. In Rayleigh fading model, the expected throughput is described in terms of well known Marcum Q-function, and such maximization is done with the aid of previous studies in liteerature on Marcum Q-function.
 
%%%%%%%%%%%%%%%%%%%%%
%%%%%%%%%%%%%%%%%%%%
%\renewcommand{\baselinestretch}{1.2}
\bibliographystyle{IEEEtran}
\bibliography{bae}

\begin{thebibliography}{10}
\providecommand{\url}[1]{#1}
\csname url@samestyle\endcsname
\providecommand{\newblock}{\relax}
\providecommand{\bibinfo}[2]{#2}
\providecommand{\BIBentrySTDinterwordspacing}{\spaceskip=0pt\relax}
\providecommand{\BIBentryALTinterwordstretchfactor}{4}
\providecommand{\BIBentryALTinterwordspacing}{\spaceskip=\fontdimen2\font plus
\BIBentryALTinterwordstretchfactor\fontdimen3\font minus
  \fontdimen4\font\relax}
\providecommand{\BIBforeignlanguage}[2]{{%
\expandafter\ifx\csname l@#1\endcsname\relax
\typeout{** WARNING: IEEEtran.bst: No hyphenation pattern has been}%
\typeout{** loaded for the language `#1'. Using the pattern for}%
\typeout{** the default language instead.}%
\else
\language=\csname l@#1\endcsname
\fi
#2}}
\providecommand{\BIBdecl}{\relax}
\BIBdecl

\bibitem{SeToBa09}
S.~Sesia, I.~Toufik, and M.~Baker, Eds., \emph{{LTE}-the {UMTS} long term
  evolution: from theory to practice}.\hskip 1em plus 0.5em minus 0.4em\relax
  West Sussex, UK: John Wiley $\&$ Sons, 2009.

\bibitem{KaSa00}
T.~Kailath, A.~H. Sayed, and B.~Hassibi, \emph{Linear estimation}.\hskip 1em
  plus 0.5em minus 0.4em\relax Upper Saddle River, NJ: Prentice Hall, 2000.

\bibitem{ToQu08}
H.~Touheed, A.~Quddus, and R.~Tafazolli, ``Predictive {CQI} reporting for
  {HSDPA},'' in \emph{Proceedings of IEEE PIMRC}, Cannes, France, Sep. 2008.

\bibitem{CuLu10}
T.~Cui, F.~Lu, V.~Sethuraman, A.~Goteti, S.~P. Rao, and P.~Subrahmanya, ``First
  order adaptive {IIR} filter for {CQI} prediction in {HSDPA},'' in
  \emph{Proceedings of IEEE WCNC}, Sidney, Australia, Apr. 2010.

\bibitem{WaZh11}
F.~Wang, T.~Zhang, C.~Feng, and R.~Li, ``Limited feedback scheme in the
  presence of feedback delay using {K}alman filter,'' in \emph{Proceedings of
  IET ICCTA}, Beijing, China, Oct. 2011.

\bibitem{TsVi05}
D.~Tse and P.~Viswanath, \emph{Fundamentals of wireless communication}.\hskip
  1em plus 0.5em minus 0.4em\relax Cambridge, UK: Cambridge University Press,
  2005.

\bibitem{BoVa09}
S.~Boyd and L.~Vandenberghe, \emph{Convex optimization}.\hskip 1em plus 0.5em
  minus 0.4em\relax Cambridge, UK: Cambridge university press, 2009.

\bibitem{CoFe02}
G.~E. Corazza and G.~Ferrari, ``New bounds for the {M}arcum {Q}-function,''
  \emph{{IEEE} Trans. Inf. Theory}, vol.~48, pp. 3003--3008, Nov. 2002.

\bibitem{ChCoMi11}
S.~Chang, P.~C. Cosman, and L.~B. Milstein, ``Chernoff-type bounds for the
  {G}aussian error function,'' \emph{{IEEE} Trans. Commun.}, vol.~59, pp.
  2939--2944, Nov. 2011.

\bibitem{Fo99}
G.~B. Folland, \emph{Real Analysis: Modern Techniques and Their
  Applications}.\hskip 1em plus 0.5em minus 0.4em\relax New York, NY, USA: John
  Wiley $\&$ Sons, 1999.

\bibitem{Ch85}
D.~Chase, ``Code combining-a maximum-likelihood decoding approach for combining
  an arbitrary number of noisy packets,'' \emph{{IEEE} Trans. Commun.},
  vol.~33, pp. 385--393, May 1985.

\bibitem{LiCoMi84}
S.~Lin, D.~Costello, and M.~Miller, ``Automatic-repeat-request error-control
  schemes,'' \emph{{IEEE} Commun. Mag.}, vol.~22, pp. 5--17, Dec. 1984.

\bibitem{Be05}
D.~Bertsekas, \emph{Dynamic programming and optimal control}.\hskip 1em plus
  0.5em minus 0.4em\relax Belmont, MA, USA: Athena Scientific, 2005.

\end{thebibliography}

\end{document}